\documentclass[11pt]{elsarticle}

\usepackage[latin1]{inputenc}
\usepackage[OT1]{fontenc}
\usepackage{amsmath}
\usepackage{amssymb}
\usepackage{amscd}
\usepackage{amsthm}
\usepackage[usenames,dvipsnames]{xcolor}
\usepackage{textcomp}
\usepackage{tabularx}

\journal{Theoretical Computer Science}

\newtheorem{prop}{Proposition}
\newtheorem{lem}{Lemma}
\newtheorem{thm}{Theorem}

\definecolor{MNred}{RGB}{100,20,20}

\begin{document}

\setlength{\parindent}{0pt}

\begin{frontmatter}
  \title{About non-monotony in Boolean automata networks}
  
  \author{Mathilde Noual$^{1,3}$, Damien Regnault$^{2}$, Sylvain
    Sen{\'e}$^{2,3,\star}$}

  \address{$^1~$Universit{\'e} de Lyon, {\'E}NS-Lyon, LIP, CNRS UMR 5668, 69007
    Lyon, France\\$^2~$Université d'{\'E}vry -- Val d'Essonne, IBISC, {\'E}A
    4526, 91000 {\'E}vry, France\\$^3~$Institut rh{\^o}ne-alpin des syst{\`e}mes
    complexes, IXXI, 69007 Lyon, France\\$^\star~$Corresponding author:
    \texttt{sylvain.sene@ibisc.univ-evry.fr}}
  
  \begin{abstract}
    This paper aims at setting the keystone of a prospective theoretical study
    on the role of non-monotone interactions in biological regulation networks.
    Focusing on discrete models of these networks, namely, Boolean automata
    networks, we propose to analyse the contribution of non-monotony to the
    diversity and complexity in their dynamical behaviours. More precisely, in
    this paper, we start by detailing some motivations, both mathematical and
    biological, for our interest in non-monotony, and we discuss how it may
    account for phenomena that cannot be produced by monotony only. Then, to
    build some understanding in this direction, we propose some preliminary
    results on the dynamical behaviour of some specific non-monotone Boolean
    automata networks called \textsc{xor} circulant  networks.
  \end{abstract}

\end{frontmatter}

\section{Introduction}

The introduction of Boolean automata networks by McCulloch and Pitts
in~\cite{McCulloch1943} and Kauffman in~\cite{Kauffman1969a,Kauffman1969b} has
initiated many developments in the study of discrete dynamical systems at the
frontier of biology, mathematics and theoretical computer science. In the
context of modelling biological regulation networks, the pertinence of abstract
networks was deeply motivated by Hopfield and Kauffman in the respective
contexts of neural and genetic networks. Among other things, Hopfield showed
in~\cite{Hopfield1982,Hopfield1984} that threshold Boolean automata networks
allow to highlight the fundamental neural concepts of associative memory and
learning.  In~\cite{Kauffman1971,Kauffman1993}, on the basis of the
breakthroughs of Jacob and Monod~\cite{Jacob1961a,Jacob1961b}, Kauffman put
emphasis on the Boolean nature of genes that are simply either actively
transcribing or not. These works as well as
Thomas'~\cite{Thomas1973,Thomas1981,Thomas1991} placed formal approaches at the
centre of the understanding of dynamical behaviours and complexity in
biology. In particular, both these works claimed that theoretical frameworks
would certainly allow biologists to bypass the observational knowledge which
cannot, alone, lead to general conclusions. Since then, numerous theoretical
studies have been carried out to acquire a better understanding of these
networks, from the computational complexity
standpoint~\cite{Floreen1989,Cosnard1992,Koiran1993,Orponen1997,Gajardo2002} as
well as from the standpoint of the characterisation of their dynamical
behaviours~\cite{Robert1986,Goles1981,Goles1990,Aracena2004a,Aracena2004b,Remy2008,Demongeot2010,Demongeot2011,Richard2011}.\medskip

In the lines of these studies and in order to complete them, we propose in this
paper to tackle the question of the role of non-monotony in Boolean automata
networks. This question seems to be missing in classical literature dealing with
Boolean automata networks as models of biological networks and, in particular,
as models of  genetic regulation networks. Indeed, on the one hand, the
underlying interaction structure of Boolean models of genetic regulation
networks are often represented by \emph{signed} digraphs where vertices
represent genes and arcs, which are labelled either by a plus or a minus sign,
represent directed actions of genes on one another, either activations or
inhibitions.  This way,  a gene that tends to influence the
expression of another gene is supposed to be either one of its activators or one
of its inhibitors. It cannot be both. That is, it cannot act as an activator
under some circumstances and act as an inhibitor under some others. This
interpretation of gene regulations leads to define monotone Boolean automata
networks as studied
in~\cite{Goles1985,Cosnard1997,Remy2003,Chaouiya2004,Colon-Reyes2005,Jarrah2010}
and~\cite{Mendoza1998,Mendoza1999,Aracena2006,Georgescu2008,Mendoza2010} from
theoretical and applied points of view respectively. On the other hand, the
class of linear networks has also been studied. This class contains in
particular the special non-monotone networks in which \emph{all} local functions
are \textsc{xor} functions. In~\cite{Cull1971}, Cull based his study
on~\cite{McCulloch1943,Huffman1956,Elspas1959} and developed an algebraic
description of the dynamical behaviour of linear
networks. In~\cite{Snoussi1980}, Snoussi gave a characterisation of behaviours
of very specific \textsc{xor} networks.  But the global dynamical properties of
general non-monotone networks have not yet been studied nor has the impact of
non-monotone interactions yet been examined {\it per se}.\medskip

Our recent studies on Boolean automata networks have however brought us to
believe that non-monotony may be one of the main causes of singular behaviours
of Boolean automata networks.  Thus, this research axis seems very pertinent in
the context of biological regulation networks. The present paper provides the
grounds of a \emph{prospective study} on non-monotony in networks. In this
context, we develop two lines. First, with some examples, we give some insights
that support the importance of non-monotony and the idea that it may be
responsible for peculiar network dynamical behaviours. Second, to serve as a
tangible starting point and build intuition, we present some primary results
concerning a particular class of non-monotone networks called \textsc{xor}
circulant networks.\medskip

In Section~\ref{sec_preliminaries}, we provide general definitions and notations
about Boolean automata networks that are crucial for the
sequel. Section~\ref{sec_motivations} details why we believe that non-monotony
is in some sense at the centre of the existence of specific dynamical
behaviours. Section~\ref{sec_xor} presents preliminary results concerning
 \textsc{xor} circulant networks. In particular, it gives properties of their
trajectorial behaviours by focusing on convergence times and of their
asymptotic behaviours by characterising attractors. Eventually,
Section~\ref{sec_conclusion} proposes perspectives to this first work on
non-monotony.

\section{Preliminary elements on Boolean automata networks}
\label{sec_preliminaries}

Informally, a \emph{Boolean automata network} involves interacting elements
whose states, which either equal $0$ (inactive) or $1$ (active), may change over
time under the influence of the states of other network
elements~\cite{Robert1986,Choffrut1988}.  This section formalises this
description by presenting the main definitions and notations which are used in
this paper.

\subsection{Structure and local transition functions}
\label{sec_structure_and_function}

A \emph{Boolean automata network} $N$ of size $n$ is composed of $n$ elements
called automata which are, by convention here, numbered from $0$ to $n-1$. For
any automaton $i$ of $V = \{0, \ldots, n-1\}$, the set of possible \emph{states}
$x_i$ of $i$ is $\{0,1\}$. Let us assume that the \emph{time space}
$\mathcal{T}$ is discrete, \emph{i.e.}, $\mathcal{T} = \mathbb{N}$. A
\emph{configuration} of $N$ corresponds to the allocation of a value of
$\{0,1\}$ to every automaton of $N$. It can thus be represented by a vector $x =
(x_0, \ldots, x_{n-1}) \in \{0,1\}^n$ and  $\{0,1\}^n$ is then the
\emph{configuration space} of $N$. Abusing language, we will denote by $x(t)$
(resp. $x_i(t)$) the configuration of $N$ (resp. the state of automaton $i$) at
time step $t \in \mathcal{T}$. Given an arbitrary configuration $x \in
\{0,1\}^n$, the \emph{density} of $x$ is defined as $d(x) =
\frac{1}{n}\cdot|\{x_i\ |\ (i \in V) \land (x_i = 1)\}|$. In our context, we
focus particularly on switches of automata states starting in a given network
configuration. For this reason, the following notations for network
configurations will be useful:
\begin{multline}
  \label{eq_notations_basiques}
  \forall x = (x_0, \ldots, x_{n-1}) \in \{0,1\}^n, \forall i \in V = \{0,
  \ldots, n-1\},\\ 
    \overline{x}^i = (x_0, \ldots, x_{i-1}, \neg x_i, x_{i+1},
  \ldots, x_{n-1})\\ \text{and, }\forall W\subseteq
  V,\ \overline{x}^{W\cup\{i\}}=\overline{\overline{x}^W}^i \text{.}
\end{multline}
Thus, in particular, $\overline{0}^i$ where $i\in V$ (resp. $\overline{0}^W$
where $W\subseteq V$) denotes the network configuration in which automaton $i$ has
state $1$ (resp. all the automata belonging to $W$ have state $1$) and all other
automata have state $0$.  The underlying interaction structure of $N$ can be
represented by a digraph $G = (V, A)$, called the \emph{interaction graph} of
$N$. In this digraph, $V$ equals the set of automata of $N$. $A \subseteq V
\times V$ is the interaction set. For any automata $i,j\in V$, it satisfies
$(j,i) \in A$ if and only if $j$ effectively influences $i$, that is, in some
network configurations (but not necessarily in all of them), the state of $j$
may cause a change of states of $i$ (see Equation~\ref{eq_effectivite}
below). As an example, Figure~\ref{fig_structure}~(left) pictures the
interaction graph of a Boolean automata network of size $3$, where $A =
\{(0,0),(0,1),(0,2),(1,0),(2,0),(2,1)\}$.
\begin{figure}[t!]
  \centerline{\includegraphics[scale=0.8]{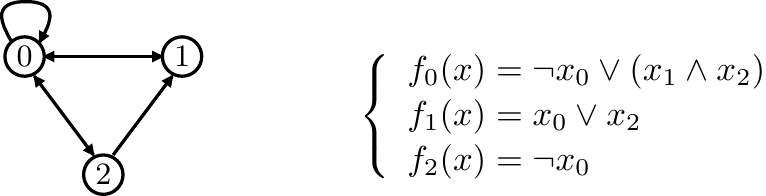}}
  \caption{(left) An interaction graph of a Boolean automata network of size $3$
    and (right) the local transition functions of its automata.}
  \label{fig_structure}
\end{figure}
Interaction graphs specify what influences apply to each automaton of a network
$N$. The nature of these influences are defined by the \emph{local transition
  functions} $f_i: \{0,1\}^n \to \{0,1\}$ which are associated to each automaton
$i$ of $N$ (as in Figure~\ref{fig_structure}~(right)) such that:
\begin{equation}
  \label{eq_effectivite}
  \exists x \in \{0,1\}^n,\ f_i(x) \neq f_i(\overline{x}^j)\ \iff\ 
  (j,i) \in A\text{.}
\end{equation}
Thus, a Boolean automata network is entirely defined by the set of local
transition functions of its automata.

\subsection{Updating modes and transition graphs}
\label{sec_updating_modes}

To determine the possible \emph{behaviours} of a network, it remains to be
specified how automata states are updated over time. The most general point of
view consists in considering all possibilities. That is, assimilating networks
with state transition systems, in each configuration, $2^n-1$ transitions are
considered, one for each non-empty set of automata whose states can be
updated. More precisely, $\forall W \neq \emptyset \subseteq V$, we define the
update function $F_W: \{0,1\}^n \to \{0,1\}^n$ such that:
\begin{equation*}
  \forall x \in \{0,1\}^n, \forall i \in V,\ F_{W}(x)_i = \begin{cases}
    f_i(x) & \text{if } i \in W \text{,}\\ 
    x_i & \text{otherwise.}
  \end{cases}
\end{equation*}
Then, according to the most \emph{general updating mode}, the global network
behaviour is given by the \emph{general transition graph} $\mathcal{G}_g =
(\{0,1\}^n, T_g)$ where $T_g = \{(x, F_W(x))\ |\ x \in \{0,1\}^n,\ W \neq
\emptyset\, \subseteq V\}$~\cite{Noual2011a,Noual2011b,Noual2011c}. In this
graph which usually is a multigraph, arcs can be labelled by the set $W$ of
automata that are updated in the corresponding transition $(x, F_W(x))$. For the
sake of clarity, in the examples of this paper, arcs with identical extremities
are represented by a unique arc with several
labels. Figure~\ref{fig_asynchronous-general}~(top) depicts the general
transition graph of the network presented in Figure~\ref{fig_structure}.\medskip
\begin{figure}[t!]
  \centerline{
    \begin{tabular}{c}
      \includegraphics[scale=0.75]{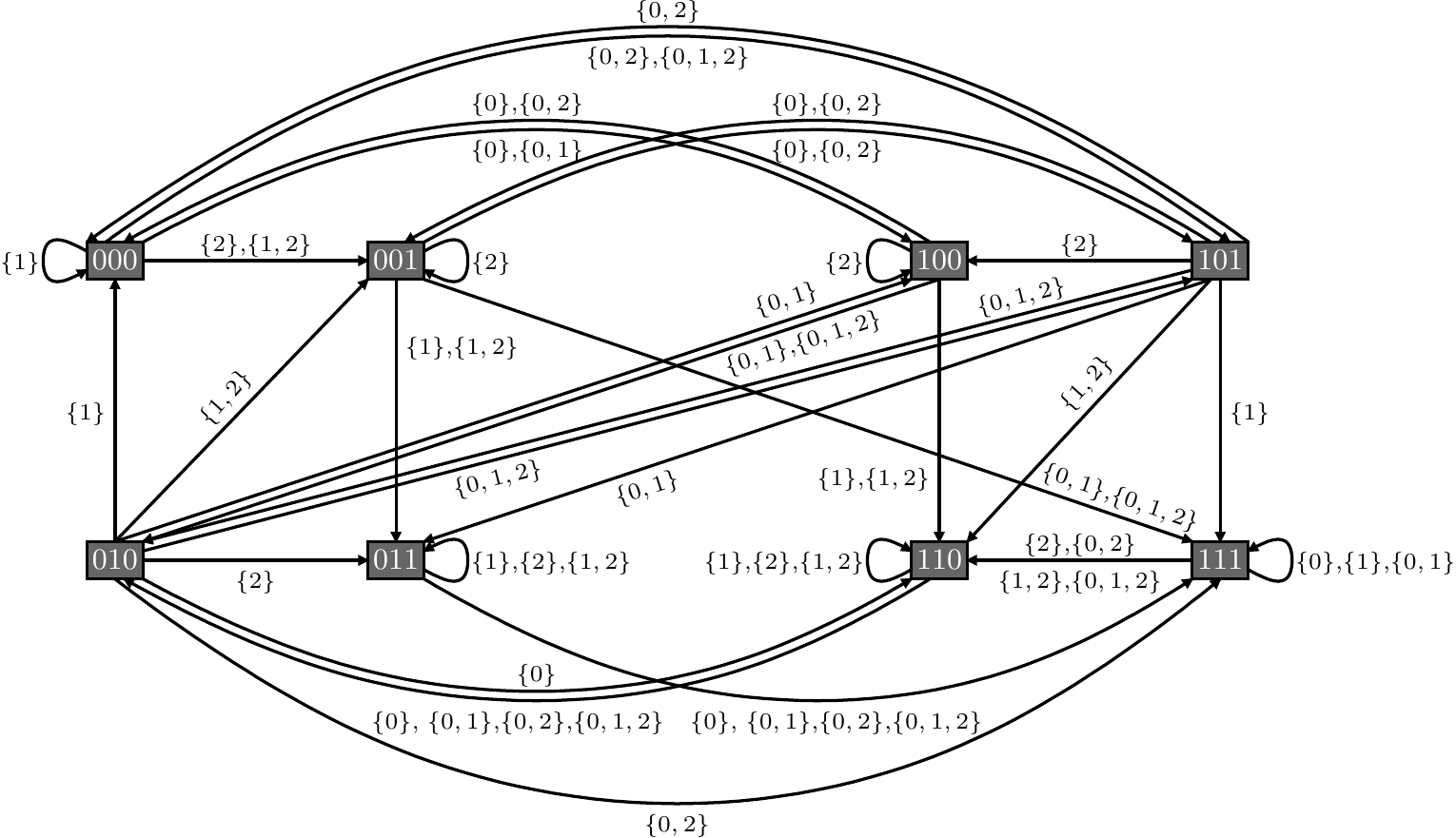}\\[2mm]
      \hline\\[-1mm]
      \includegraphics[scale=0.75]{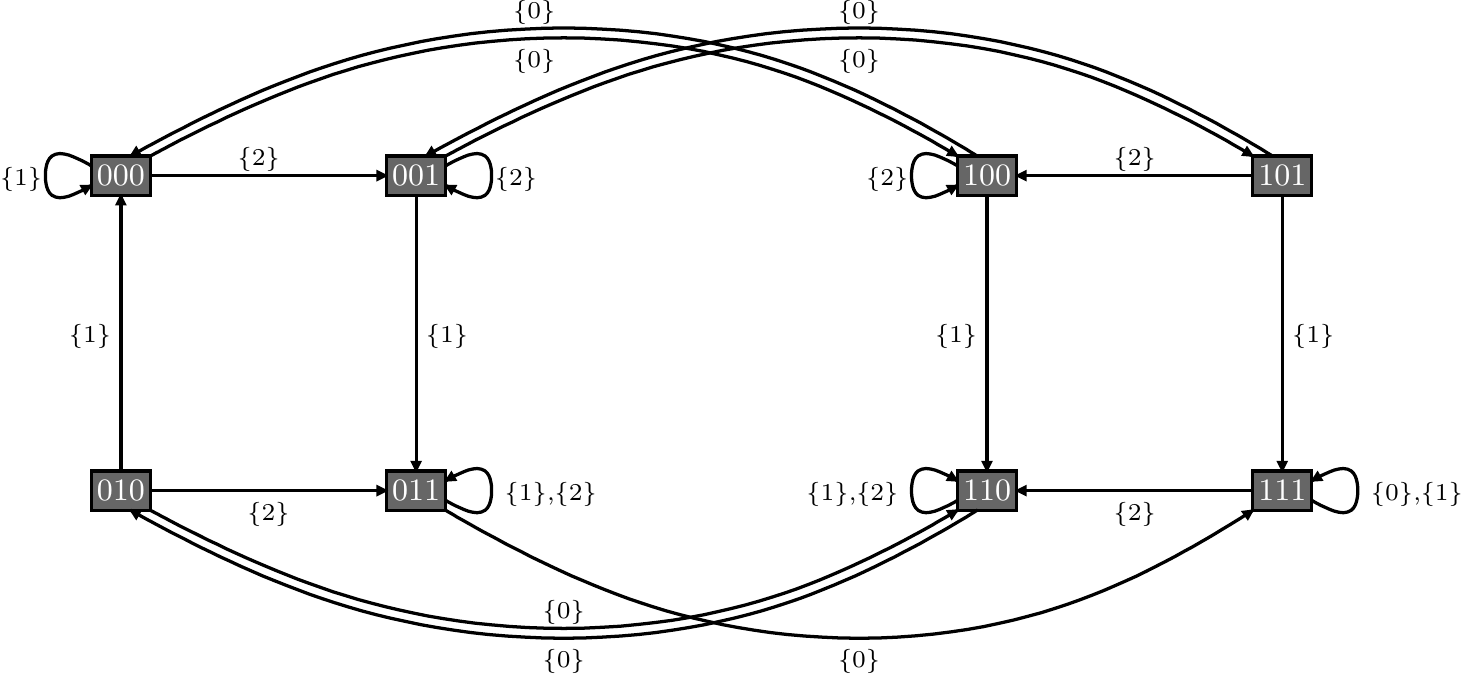}\\
    \end{tabular}
  }
  \caption{(top) General and (bottom) asynchronous transition graphs of the
    Boolean automata network of Figure~\ref{fig_structure}.}
  \label{fig_asynchronous-general}
\end{figure}

Transitions $(x, F_{i}(x))$ that only involve the update of one automaton $i \in
V$ are called \emph{asynchronous transitions}. Transitions $(x, F_W(x)), |W| >
1$ that involve the update of several are called \emph{synchronous transitions}.
The subgraph $\mathcal{G}_a = (\{0,1\}^n, T_a)$ of $\mathcal{G}_g$ whose set of
arcs $T_a = \{(x, F_{\{i\}}(x))\ |\ x \in \{0,1\}^n,\ i \in V\}$ equals the set
of \emph{asynchronous transitions} of the network is called the
\emph{asynchronous transition graph}. This graph defines the \emph{asynchronous
  updating mode} according to which, in each configuration, only $n$ transitions
are considered, one for each automaton that can be updated alone. This updating mode
has been widely used in studies of Thomas and his co-workers
in~\cite{Thomas1981,Thomas1991,Remy2003,Richard2004,Richard2007,Remy2008}.  An
illustration of an asynchronous transition graph is given in
Figure~\ref{fig_asynchronous-general}~(bottom).\medskip

Because both the general and the asynchronous transition graphs are very large
graphs, in some cases, to draw some intuitions, it may be necessary to restrict
our attention to the transitions that are allowed under a specific deterministic
updating schedule $u$.  This amounts to considering a transition graph
$\mathcal{G}_{u} = (\{0,1\}^n, T_{u})$ which is the graph of a function $F[u]:
\{0,1\}^n \to \{0,1\}^n$ ({\it i.e.}, $T_u = \{(x, F[u](x))\ |\ x \in
\{0,1\}^n\}$). This function has the following form: $F[u]= F_{W_{p-1}}\circ
\ldots \circ F_{W_1}\circ F_{W_0}$ where $p\in \mathbb{N}$ and $\forall k\leq
p,\ W_k\subseteq V$. It is called the \emph{global transition function}
associated to the updating schedule $u$, that updates simultaneously all
automata in $W_0$, then updates simultaneously all automata in $W_1$ \ldots This
second point of view has been adopted
in~\cite{Demongeot2008,Elena2008,Aracena2009,Aracena2010,Goles2010,Goles2011}
following the introduction of \emph{block-sequential updating schedules} by
Robert in~\cite{Robert1986,Robert1995} (see Figure~\ref{fig_block-sequential}).
Section~\ref{sec_xor} is set in similar lines. It focuses on the \emph{parallel
  updating mode} $\pi$ which consists in deterministically updating all network
automata at once in each network configuration. In this case, the global
transition function is $F[\pi] = F_V$ such that $\forall i\in V,\ F[\pi](x)_i =
f_i(x)$ and the network behaviour is considered to be described by the graph of
$F[\pi]$, that is, the transition graph $\mathcal{G}_{\pi} = (\{0,1\}^n,
T_{\pi})$ where $T_{\pi} = \{(x, F[\pi](x))\ |\ x \in \{0,1\}^n\}$.\medskip
\begin{figure}[t!]
  \centerline{ 
    {\footnotesize \begin{tabular}{ccccc}
        \includegraphics[scale=0.75]{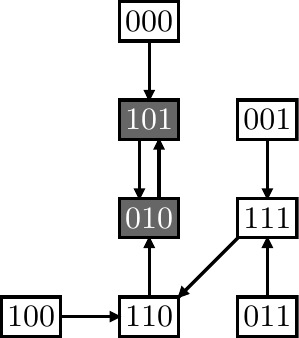} & \hspace*{3mm} &
        \includegraphics[scale=0.75]{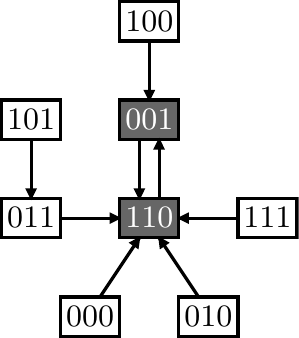} & \hspace*{3mm} &
        \includegraphics[scale=0.75]{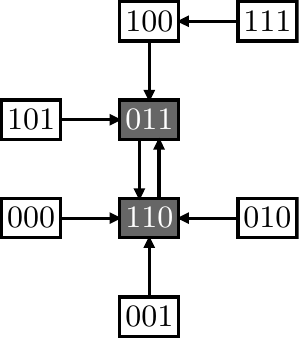}\\
        \centering $(a)$ & & \centering $(b)$ & & \centering $(c)$
      \end{tabular}}
  }
  \caption{Transition graphs of the Boolean automata network of
    Figure~\ref{fig_structure} associated to $(a)$ the parallel updating mode,
    $(b)$ a block-sequential updating schedule $u$ whose global transition
    function is $F[u]=F_{\{1,2\}}\circ F_{\{0\}}$ and $(c)$ a particular
    block-sequential updating schedule $s$, called sequential, whose global
    transition function is $F[u]=F_{\{1\}}\circ F_{\{2\}}\circ
    F_{\{0\}}$. Recurrent configurations appear in grey.}
  \label{fig_block-sequential}
\end{figure}

\subsection{Dynamical behaviours and non-monotony}
\label{sec_non-monotony}

Let $N$ be an arbitrary Boolean automata network. Consider any updating mode $u$
among those defined above and let $\mathcal{G}_u$ be the corresponding
transition graph. Let $x \in \{0,1\}^n$ be a configuration of $N$. We call
\emph{trajectory} of $x$ any path in $\mathcal{G}_u$ that starts in
$x$. Terminal strongly connected components of $\mathcal{G}_u$ are called the
\emph{attractors} of $N$ and constitute the asymptotic behaviours of $N$. Their
size equals the number of configurations that they contain.  Configurations that
belong to an attractor are called \emph{recurrent configurations}. Attractors of
size $1$ (resp. of size strictly greater than $1$) are called \emph{stable
  configurations} (resp. \emph{stable oscillations}). When $\mathcal{G} =
\mathcal{G}_u$ is the transition graph associated to a deterministic updating
schedule $u$, stable configurations correspond to fixed points of the global
transition function $F[u]$ and stable oscillations of size $p$, which are rather
called \emph{limit cycles} of \emph{period} $p$ in this case, correspond to
oriented cycles in $\mathcal{G}_u$. As an example,
Figure~\ref{fig_block-sequential} shows that the network of
Figure~\ref{fig_structure} admits one unique attractor, a limit cycle of period
$2$, under any of the three deterministic updating schedules considered. The
precise definition of this limit cycle, however, differs in each case.  In
particular, as proven in~\cite{Aracena2009,Goles2008}, no configurations besides
stable configurations are recurrent under the parallel updating schedule as well
as under the sequential updating schedule which updates one automaton at the
time been.  Furthermore, Figure~\ref{fig_asynchronous-general} shows that the
same network admits one unique attractor, a stable oscillation of size $8$,
indifferently when it is subjected to the asynchronous or general updating
modes.  \medskip

By analogy with continuous functions, the local transition function $f_i$ of an
automaton $i \in V$ is said to be \emph{locally monotone} in $j \in V$ if,
either:
\begin{multline*}
 \forall x = (x_0, \ldots, x_{n-1}) \in \{0,1\}^n,\\ f_i(x_0, \ldots, x_{j-1}, 0, x_{j+1}, \ldots, x_{n-1}) \leq f_i(x_0, \ldots,
  x_{j-1}, 1, x_{j+1}, \ldots, x_{n-1})
\end{multline*}
or:
\begin{multline*}
\forall x = (x_0, \ldots, x_{n-1}) \in \{0,1\}^n,\\  f_i(x_0, \ldots, x_{j-1}, 0, x_{j+1}, \ldots, x_{n-1}) \geq f_i(x_0, \ldots,
  x_{j-1}, 1, x_{j+1}, \ldots, x_{n-1}) \text{.}
\end{multline*}
In other terms, $f_i$ is locally monotone in $j$ if, in the conjunctive normal
form of $f_i(x)$, either only $x_j$ appears or only $\neg x_j$ does.  The
function $f_i$ is said to be locally monotone or simply monotone if it is
locally monotone in \emph{all} $j\in V$.  It is said to be \emph{non}
(\emph{locally}) \emph{monotone} otherwise. In this latter case, there exists $j
\in V$ such that in some configurations, the state of $i$ tends to imitate that
of $j$ and in some other configurations, on the contrary, the state of $i$ tends
to negate that of $j$. When all functions $f_i$,  $i \in V$, are monotone,
the network is said to be \emph{monotone} itself. Otherwise, if at least one
local transition function is non-monotone, the network is said to be
\emph{non-monotone}.

\section{Motivations}
\label{sec_motivations}

To put forward the importance of studying non-monotony in discrete models of
regulation networks, let us first recall the fundamental concept of genetics
establishing that a gene is a portion of the \textsc{dna} which is transcribed
into a m\textsc{rna} (the gene is then said to be expressed) that is itself
translated into one or several proteins, called the \emph{products} of that
gene. Because proteins can influence the transcription and translation stages,
genes have the possibility of interacting with one another through their
products. Further, because the effect of a protein may depend on its
concentration in the cell, genes may have different effects on one another. If
gene \textsf{g$_{\textsf{j}}$} influences the expression of gene
\textsf{g$_{\textsf{i}}$} via one of its protein products \textsf{p}, then it
may do so differently according to the concentration of \textsf{p} in the cell.
As an illustration, let us consider the infection of a bacterium
\emph{Escherichia coli} by a phage $\lambda$~\cite{Lederberg1950} at high
temperature. The genetic regulations that allow a phage $\lambda$ to enter its
lysogenic and lytic cycles\footnote{The lysogenic cycle of a phage $\lambda$ is
  the stage where its genome is inserted in the genome of the bacterium. Its
  lytic cycle is the stage where it replicates, leading \emph{in fine} to the
  death of the bacterium.}  involve two genes, \textsf{Cro} and
\textsf{cI}. Indeed, at high temperature, \textsf{Cro} influences itself and the
nature of this influence is different according to the concentration of its
protein product \textsf{p$_{\textsf{Cro}}$}. In~\cite{Eisen1970,Thieffry1995},
the authors show that if the \textsf{p$_{\textsf{Cro}}$} concentration in the
cell is low (resp. high), \textsf{Cro} tends to activate (resp. inhibit) its own
expression, whatever the \textsf{p$_{\textsf{cI}}$} concentration is. This
induces an increase (resp. a decrease) of the \textsf{p$_{\textsf{Cro}}$}
concentration. Supposing that automata networks are reasonable models of genetic
networks, such duality in the influence of one gene \textsf{g$_{\textsf{j}}$} on
the state of the same or another gene \textsf{g$_{\textsf{i}}$} corresponds
precisely to the formal notion of non-monotony, specifically, that of $f_i(x)$
with respect to $x_j$.\medskip

In the light of some recent developments \cite{Noual2011a,Delaplace2011}, the
study of non-monotony in discrete networks also seems pertinent from a
different, mathematical standpoint. Indeed, first, let us consider the problem
of modularity in gene regulation networks which is essential in the context of
biology. Modules are informally defined as independent groups of interacting
genes. More precisely, they involve minimal sets of genes (or biobricks) that
own independent behaviours specific to real biological functions. Until now, the
notion of modularity used in the literature relies on structural
parameters~\cite{Milo2002,Rives2003,Gagneur2004}. For instance, modules are
often defined simply as strongly connected components. However, although such
structural definitions are natural, they can only lead to structural results,
necessarily failing to reveal biobricks. In~\cite{Delaplace2011}, a new notion
of modularity is introduced for asynchronous multi-state automata networks
considered as discrete models of gene regulation networks. It is supported by
dynamical considerations that take into account the asymptotic behaviours of
these networks. To detail this, let us consider a network $N$ whose associated
interaction graph is $G = (V,A)$. A modular organisation of a network is defined
in~\cite{Delaplace2011} by an ordered partition%
\footnote{An \emph{ordered partition} $(S_0, \ldots, S_{k-1})$ of an arbitrary
  set $S$, is defined by $k \leq |S|$ non-empty ordered subsets $S_i\subset S$
  such that $S = \biguplus_{i<k} S_i$ (\emph{i.e.}, $\forall i, j<k,\ i\neq j
  \Rightarrow S_i \cap S_j = \emptyset$ and $S = \bigcup_{i<k} S_i$).} $(V_0,
\ldots, V_{k-1})$ of the set of network automata. 
The subgraph of $G$ induced by $V_i$ is said to be a module of $N$ if the
ordered sets $V_i$, $i<k$, satisfy $\mathcal{E} = \mathcal{E}_0 \oslash \ldots
\oslash \mathcal{E}_{k-1}$, where $\mathcal{E}_W$, $W\subseteq V$ denotes the
asymptotic behaviour of the subnetwork of $N$ induced by $W$ and $\oslash$ is a
composition operator defined in~\cite{Delaplace2011}.  In this context, the
authors of~\cite{Delaplace2011} show that any topological ordering\footnote{If
  $G = (V,A)$ is a digraph, a \emph{topological ordering} of $G$ is a linear
  ordering of the vertices in $V$ such that, $\forall (i,j) \in A$, $i$ comes
  before $j$ in the ordering.} on the set of strongly connected components of
the underlying structure of a network $N$ does indeed define a modular
organisation of $N$. However, in the general case, strongly connected components
are not minimal modules and thus do not allow to reveal biobricks. In some
cases, they can be decomposed into smaller independent sub-modules. What is
interesting is that for Boolean automata networks, all encountered examples of
non-decomposable strongly connected components involve non-monotony.\medskip

In different lines, the importance of non-monotony can also be seen by adapting
some results presented in~\cite{Noual2011a}. This produces
Proposition~\ref{prop_asynchronous-general} below which relates non-monotony to
non-trivial changes in the dynamical behaviours of a network when synchronism
is added to its asynchronous behaviour. In this proposition, a transition $(x,
F_W(x))$ is said to be \emph{sequentialisable} if there exists a series of
consecutive asynchronous transitions $(x, F_{\{i\}}(x)), (F_{\{i\}}(x),
F_{\{j\}} \circ F_{\{i\}}(x)), \ldots$ that start in $x$ and end in $F_W(x)$.
\begin{figure}[t!]
  \centerline{\small
    \begin{tabular}{c}
      \includegraphics[scale=0.8]{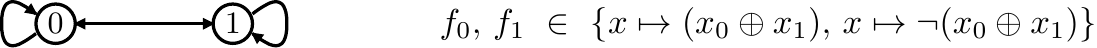}\\[2mm]
      \hline
    \end{tabular}
  }\vspace*{4mm}
  \centerline{
    \begin{tabular}{c|c}
      \includegraphics[scale=0.75]{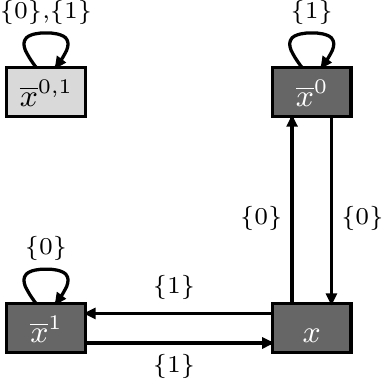}~~~~~ &
      ~~~~~\includegraphics[scale=0.75]{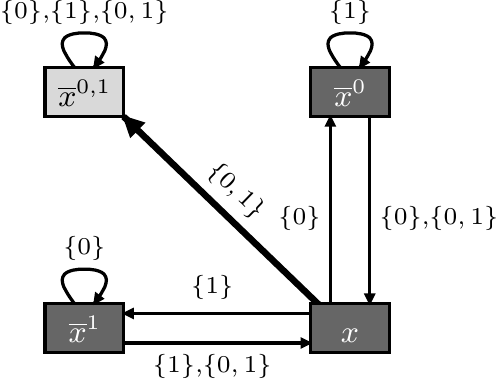}
    \end{tabular}
  }
  \caption{Top panel: Generic description of the four smallest Boolean automata
    networks that satisfy the conditions of
    Proposition~\ref{prop_asynchronous-general}. Bottom panels: (left) generic
    (see proof of Proposition~\ref{prop_asynchronous-general}) asynchronous and
    (right) general transition graphs of these networks.}
  \label{fig_motiv}
\end{figure}
\begin{prop}
  \label{prop_asynchronous-general} 
  The smallest Boolean automata networks that have non-sequentialisable
  synchronous transitions and significantly different limit behaviours under the
  asynchronous and general updating modes are non-mono\-tone.
\end{prop}
\begin{proof} 
  Let us find the smallest network $N$ with a non sequentialisable synchronous
  transition. Obviously, this network needs to have more than one automaton and
  if it has size $2$, then, to have a non-sequentialisable synchronous
  transition, its general transition graph needs to contain a subgraph of the
  following form:\\[2mm] \centerline{
    \includegraphics[scale=0.75]{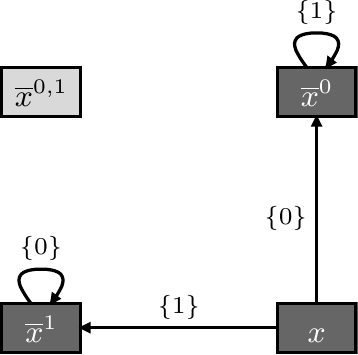} }\\[2mm] where
  $\overline{x}^{\, i,j} = \overline{x}^{\, \{i,j\}} =
  \overline{\overline{x}^i}^j$ (see
  Equation~\ref{eq_notations_basiques}). Moreover, to have significantly
  different asymptotic behaviours under the asynchronous and general updating
  modes, the synchronous transition $(x, \overline{x}^{\, i,j})$ must go out of
  a set of configurations that induces a stable oscillation under the
  asynchronous updating mode. Thus, the general transition graph of $N$ must
  have the form of the general transition graph pictured in the bottom right
  panel of Figure~\ref{fig_motiv}.  Then, only two functions $f_0$ are
  possible. If in configuration $x$ above, $x_0 = 1$, then, $f_0(x): x \mapsto
  x_0 \oplus x_1$ where $\oplus$ denotes the \textsc{xor}
  connector\footnote{$\forall a,b\in \{0,1\},\, a\oplus b= (a\land \neg b)\vee
    (\neg a \land b)$.}. If in configuration $x$ above, $x_0 = 0$, then $f_0(x)
  : x \mapsto \neg ( x_0 \oplus x_1)$. The function $f_1$ is defined
  similarly. In conclusion, there are four smallest networks satisfying the
  properties of Proposition~\ref{prop_asynchronous-general}. They have size $2$
  and their interactions graph equal the graph pictured in the top panel of
  Figure~\ref{fig_motiv}. Their two local interaction functions $f_0$ and $f_1$
  either equal $x \mapsto x_0 \oplus x_1$ or $x \mapsto \neg ( x_0 \oplus x_1)$.
\end{proof}
The reasons evoked in this section that led us to focus on non-monotony in
automata networks emphasise the apparent importance of non-monotone
functions. In the next section, to initiate an analysis of the behaviours of
general non-monotone networks and develop some intuition in this direction, we
focus on a specific class of non-monotone networks, namely \textsc{xor} circulant
networks, and study some of their dynamical properties.\medskip

\section{{\footnotesize XOR} circulant networks}
\label{sec_xor}

Before we present some results on the trajectorial and asymptotic dynamical
behaviours of \textsc{xor} circulant networks, let us first introduce some
definitions and preliminary properties in relation to these.\medskip

\subsection{Definitions and basic properties}

A \emph{circulant matrix} $\mathcal{C}$ is a matrix of order $n$ whose
$i^{\text{th}}$ row vector $\mathcal{C}_i$ ($i<n$) is the right-cyclic
permutation with offset $i$ of its first row vector $\mathcal{C}_0$ so that
$\mathcal{C}$ has the following form:
\begin{equation*}
  \mathcal{C} = \begin{pmatrix}
    c_0 & c_1 & c_2 & \ldots & c_{n-1}\\
    c_{n-1} & c_0 & c_1 & \ldots & c_{n-2}\\
    c_{n-2} & c_{n-1} & c_0 & \ldots & c_{n-3}\\
    \vdots & \vdots & \vdots & \ddots & \vdots \\
    c_1 & c_2 & c_3 & \ldots & c_{0}\\
  \end{pmatrix}\text{.}
\end{equation*}
For any integer $k \geq 2$, a \emph{$k$-\textsc{xor} circulant network} of size
$n \geq k$ is a Boolean automata network with $n$ automata that can be numbered
so that the following four properties are satisfied: $(i)$ the adjacency matrix
$\mathcal{C}$ of the network interaction graph $G = (V, A)$ ({\it i.e.}, the
$n\times n$ matrix $\mathcal{C}$ defined by $\forall i,j\in V,\,
\mathcal{C}_{i,j}=1\iff (j,i)\in A$), called the \emph{interaction matrix} for
short, is a circulant matrix, $(ii)$ each row $\mathcal{C}_{i}$, $i \in V$, of
this matrix contains exactly $k$ non-null coefficients (\emph{i.e.}, $\forall i
\in V,\ \sum_{j \in V} \mathcal{C}_{i,j} = deg^-_G(i)=k$), $(iii)$
$\mathcal{C}_{0,n-1}=c_{n-1}=1$ and $(iv)$ the local transition function $f_i$
of any automaton $i \in V$ is a \textsc{xor} function:
\begin{equation*}
  \label{eq_flt}
  \forall x\in \{0,1\}^n,\ f_i(x) = \bigoplus_{j\in V} \mathcal{C}_{i,j} \cdot x_j = 
  \sum_{j\in V} \mathcal{C}_{i,j} \cdot x_j~[2]\text{,}
\end{equation*}
where, for any integers $a$ and $b$, $a\, [b]$ stands for $a~(\text{mod}~b)$.
In the sequel, for the sake of simplicity, \textsc{xor} circulant networks are
considered to be subjected to the \emph{parallel updating mode} so that if $x =
x(t) \in \{0,1\}^n$ is the network configuration at time step $t \in
\mathcal{T}$, then the network configuration at time step $t + 1$ equals $x(t+1)
= F(x) = \mathcal{C} \cdot x$ (where operations are supposed to be taken modulo
$2$). Thus, a \textsc{xor} circulant network is completely defined by its
interaction graph $G = (V, A)$ or by its interaction matrix
$\mathcal{C}$. Figure~\ref{fig_2-xor-network} pictures two examples of
$2$-\textsc{xor} circulant networks of size $5$. Let us note that one of the
four networks satisfying Proposition~\ref{prop_asynchronous-general} and defined
in Figure~\ref{fig_motiv} is also a $2$-\textsc{xor} circulant network of size
$2$.\medskip
\begin{figure}[t!]
  \centerline{
    {\footnotesize \begin{tabular}{c}
        \includegraphics[scale=0.8]{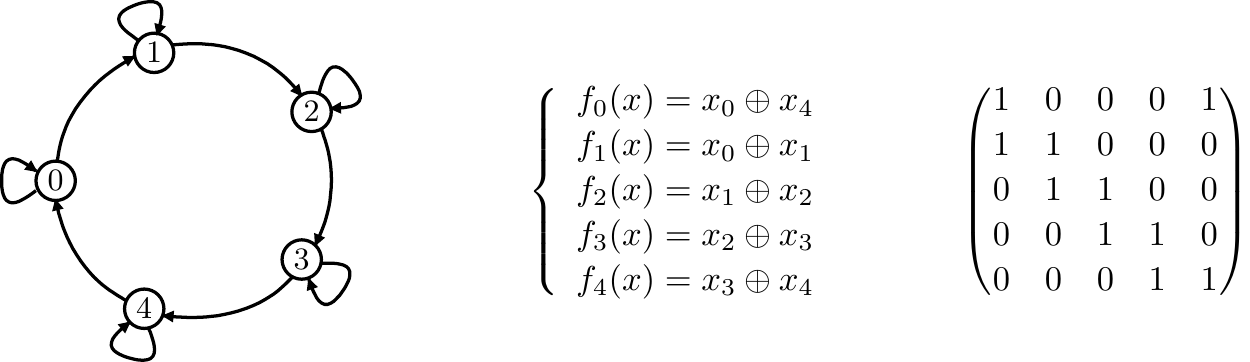}\\[2mm]
        \hline\\[-0.5mm]
        \includegraphics[scale=0.8]{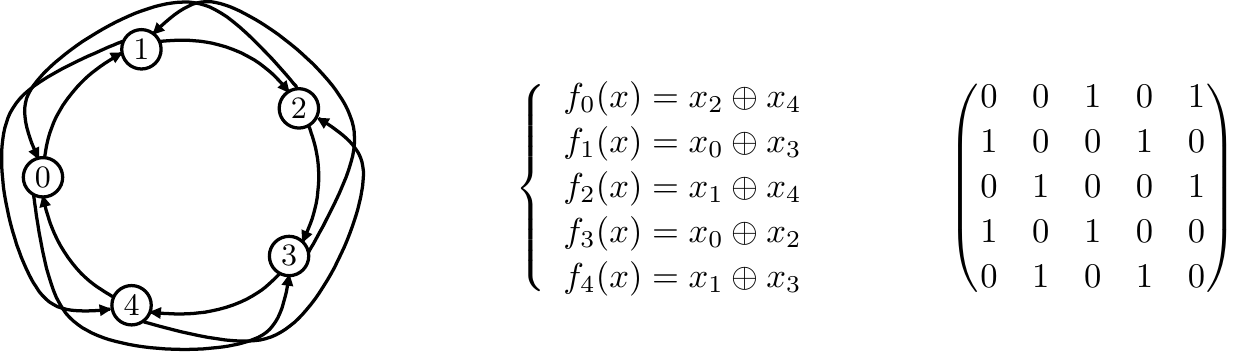} 
      \end{tabular}}
  }
  \caption{Interaction graphs, local transition functions and interaction
    matrices of two $2$-\textsc{xor} circulant networks of sizes $5$.}
  \label{fig_2-xor-network}
\end{figure}

Let us note that by point $(iii)$ in their definitions, $k$-\textsc{xor}
circulant networks have Hamiltonian circuits underlying their structures. When
automata are numbered as suggested in this definition, these circuits are
composed of the set of arcs $\{(i,i+1~[n])\ |\ i\in V\}\subseteq A$.  More
generally, it can be shown that each non-null coefficient
$c_j=\mathcal{C}_{0,j}$ of a circulant interaction matrix $\mathcal{C}$ induces
$gcd\, (n,j)$ independent circuits of length $n/gcd\, (n,j)$ in the interaction
graph $G$ of the corresponding network. Now, it has been shown that to have
several stable configurations and/or stable oscillations, Boolean automata
networks need to have circuits underlying their interaction
graphs~\cite{Thomas1981,Richard2007,Richard2010}. Thus, the presence of circuits
underlying the structures of $k$-\textsc{xor} circulant networks guarantees that
these networks have interesting, non-trivial dynamical behaviours.\medskip

Any $k$-\textsc{xor} circulant network $N$ can be seen in terms of
\emph{cellular automata}. Indeed, if $N$ has size $n$ and interaction graph $G =
(V, A)$, it can be modelled by the finite one-dimensional cellular automaton
that has $n$ cells assimilated to the $n$ automata of $N$ and that satisfies
what follows. The \emph{neighbourhood} $\mathcal{N}$ of a cell $i \in V$ equals
the in-neighbourhood of automaton $i$ in $N$: $\mathcal{N} = \{j \in V\ |\ (j,
i) \in A\}$. The local rule $\gamma: \{0,1\}^{|\mathcal{N}|} \to \{0,1\}$ of the
cellular automaton is defined similarly to the local transition functions of
$N$: $\gamma((x_\ell)_{\ell \in \mathcal{N}}) = \bigoplus_{\ell \in \mathcal{N}}
x_{\ell}$. In the sequel, we use this formalisation to exploit tools drawn from
the theory of cellular automata. Thus, if $x = x(0) \in \{0,1\}^n$ is an initial
configuration of $N$, we consider the corresponding \emph{space-time diagram},
that is, the grid of $\{0,1\}^n \times \mathcal{T}$ whose line $t \in
\mathcal{T}$ represents $x(t)$, \emph{i.e.}, the configuration of $N$ at time
step $t$. The \emph{trace} of cell or automaton $i \in V$ then corresponds to
column $i$ of this grid, that is, to the sequence $(x_i(t))_{t \in
  \mathcal{T}}$. Also, for an arbitrary configuration $x \in \{0,1\}^n$ and an
automaton $i \in V$, $S_i(x)$ denotes the configuration that satisfies $\forall
j \in V,\ S_i(x)_j = x_{2i-j~[n]}$. It is called the \emph{symmetric of $x$}
with respect to $i$. We write $\widetilde{N}$ to denote the \emph{symmetric of
  $N$}, that is, the $k$-\textsc{xor} circulant network whose interaction matrix
is $\mathstrut^t\mathcal{C}$. In the sequel, by default, $\mathcal{N}^-(i)$
(resp. $\mathcal{N}^+(i)$) denotes the \emph{in-neighbourhood} (resp. the
\emph{out-neighbourhood}) of automaton $i$ in $N$ and
$\widetilde{\mathcal{N}}^-(i)$ (resp. $\widetilde{\mathcal{N}}^+(i)$) denotes
its in-neighbourhood (resp. its out-neighbourhood) in $\widetilde{N}$. This way,
for any two automata $i, j \in V$, $j \in \mathcal{N}^-(i) \iff j \in
\widetilde{\mathcal{N}}^+(i)$. The global transition function of $\widetilde{N}$
is denoted by $\widetilde{F}$ if that of $N$ is denoted by $F$. One last
convention that is used throughout the sequel is the following. By default,
unless $N$ is the symmetric of another $k$-\textsc{xor} circulant network that
was introduced before, its automata are supposed to be numbered as suggested
above in the definition of $k$-\textsc{xor} circulant networks so that
$c_{n-1}=\mathcal{C}_{0,n-1}=1$. This way, $\{(i,i+1~[n])\ |\ i\in V\}\subseteq
A$ defines a Hamiltonian circuit in the structure of $N$ and $\{(i+1~[n], i)\ |\
i\in V\}\subseteq A$ defines a Hamiltonian circuit in the structure of its
symmetric $\widetilde{N}$.\medskip

To end this paragraph, we list some basic properties of \textsc{xor} circulant
networks that follow directly from the definitions of \textsc{xor} functions and
circular matrices:
\begin{prop}
  \label{prop_basic}
  ~
  \begin{enumerate}
  \item[1.]\label{prop_number} The number of $k$-\textsc{xor} circulant networks
    of size $n$ equals $\binom{k-1}{n-1}$.
  \end{enumerate}
  Any $k$-\textsc{xor} circulant network of size $n$ satisfies the following
  properties:
  \begin{enumerate}
  \item[2.]\label{prop_0} Configuration $(0, \ldots, 0)$ is a stable
    configuration.
  \item[3.]\label{prop_1} Configuration $(1,\ldots,1)$ is an antecedent of
    $(0,\ldots,0)$ if $k$ is even or a stable configuration if $k$ is odd.
  \item[4.]\label{prop_isomorphism} The trajectory of a configuration $x$ is
    isomorphic to that of any configuration $y$ which is a circular permutation
    of $x$.
  \end{enumerate}
\end{prop}

\subsection{Results}

\subsubsection{General $k$-\textsc{xor} circulant networks}

First, in this paragraph, we concentrate on general $k$-\textsc{xor} circulant
networks and exploit the cellular automata formalisation presented above to
derive some features of the dynamical behaviours of these networks.
\begin{lem}
  \label{mask}
  Let $N$ be a $k$-\textsc{xor} circulant network of size $n$ with automata set
  $V$ and symmetric global transition function $\widetilde{F}$. For any
  automaton $i\in V$, let $M_i(t)$, $t \in \mathcal{T}$, denote the set of
  automata which have state $1$ in configuration
  $\widetilde{F}^t(\overline{0}^i)$.  Then, $\forall x(0)\in \{0,1\}^n,\ \forall
  t \in \mathcal{T},\ x_i(t) = \bigoplus_{j \in M_i(t)} x_j(0)$.
\end{lem}  
\begin{proof} 
  We prove Lemma~\ref{mask} by induction on $t \in \mathcal{T}$.\smallskip

  For $t = 0$, $M_i(0) = \{i\}$ holds by definition of configuration
  $\overline{0}^i$. Thus, $\forall x(0) \in \{0,1\}^n,\ x_i(0) = \bigoplus_{j
    \in M_i(0)} x_j(0)$.\smallskip
  
  Now, suppose that $\forall x(0) \in \{0,1\}^n,\ x_i(t) = \bigoplus_{j \in
    M_i(t)} x_j(0)$ and consider the initial configuration $y(0) \in
  \{0, 1\}^n$.\smallskip
  
  Since $y(t+1) = \widetilde{F}^{t+1}(y(0)) = \widetilde{F}^t(y(1))$, applying
  the induction hypothesis to configuration $x(0) = y(1)$ yields:
  \begin{equation*}
  y_i(t+1) = \bigoplus_{j \in M_i(t)} y_j(1)\text{.}
  \end{equation*}
  By definition, $\forall j \in V,\ y_j(1) = f_j(y(0)) = \bigoplus_{\ell \in
    \mathcal{N}^- (j)} y_\ell(0) = \bigoplus_{\ell \in \widetilde{\mathcal{N}}^+
    (j)} y_\ell(0)$. Thus, with the commutativity and associativity of the
  $\oplus$ operator, we can derive that:
  \begin{equation*}
    \begin{split}
      y_i(t+1) & = \bigoplus_{j \in M_i(t)} \big(
      \bigoplus_{\ell \in \widetilde{\mathcal{N}}^+ (j)} y_\ell(0)\big)\\
      & = \bigoplus_{\{\ell\  |\widetilde{\mathcal{N}}^-(l) \cap
        M_i(t)| = 1~[2]\}} y_l(0)
    \end{split}
  \end{equation*}
  Now, let us remark that $\forall t\in\mathcal{T},\
  \widetilde{F}(\overline{0}^{M_i(t)}) = \overline{0}^{M_i(t+1)}$ by
  definition. Then, $\forall \ell \in V, \overline{0}^{M_i(t+1)}_\ell= 1$ if and
  only if $|\widetilde{\mathcal{N}}^-(l)\cap M_i(t)|\equiv 1~ [2]$. 
  From this follows $y_i(t+1) = \bigoplus_{j \in M_i(t+1)} y_j(0)$ and then
  $\forall t \in \mathcal{T},\ x_i(t) = \bigoplus_{j \in M_i(t)} x_j(0)$.
\end{proof}
\begin{lem}
  \label{symmetric_network} 
  Let $N$ be a $k$-\textsc{xor} circulant network of size $n$ with automata set
  $V$ and global transition function $F$. For any automaton $i\in V$, and for
  any configuration $x\in \{0,1\}^n$, it holds that $\widetilde{F}(S_i(x)) =
  S_i(F(x))$.
\end{lem}  
\begin{proof} 
  For any $j \in V$, the following holds:
  \begin{equation*}
    \begin{split}
      \widetilde{F}(S_i(x))_j & = \bigoplus_{\ell \in \widetilde{\mathcal{N}}^-
        (j)} (S_i(x))_\ell\ =\ \bigoplus_{\ell \in \widetilde{\mathcal{N}}^- (j)} 
      x_{2i-\ell~[n]}\\
      & = \bigoplus_{\{\ell \text{ s.t. } 2i-\ell~[n]\, \in\, \widetilde{\mathcal{N}}^-(j)\}}
      x_{\ell}\\
      & = \bigoplus_{\{\ell \text{ s.t. } j\, \in\, \mathcal{N}^-(2i-\ell~[n])\}}
      x_{\ell}\text{.}
    \end{split}
  \end{equation*} 
  If $j \in \mathcal{N}^-(2i-\ell~[n])$, then all automata $l, l' \in V$ of $N$
  such that $l-l' = j-(2i-\ell)~[n]$ are such that $l \in \mathcal{N}^-(l')$. In
  particular, if automaton $j \in \mathcal{N}^-(2i-\ell~[n])$, then $\ell \in
  \mathcal{N}^-(2i-j)$. Hence:
  \begin{equation*}
    \begin{split}
      \bigoplus_{\{\ell \text{ s.t. } j\, \in\, \mathcal{N}^-(2i-\ell~[n])\}} x_{\ell} &
      = \bigoplus_{\ell\, \in\, \mathcal{N}^-(2i-j~[n])} x_{\ell}\\
      & = F(x)_{2i-j}\\
      & = (S_i(F(x)))_j\text{,}
    \end{split}
  \end{equation*}
  and Lemma~\ref{symmetric_network} follows.
\end{proof}
\begin{prop}
  \label{symmetric_space_time_diagram} 
  Let $N$ be a $k$-\textsc{xor} circulant network of size $n$ with automata set
  $V$ and global transition function $F$. For any automaton $i \in V$ and for
  the initial configuration $x(0)=\overline{0}^i$, it holds that $\forall t \in
  \mathcal{T}, \widetilde{F}^t(x(0)) = S_i(x(t))$.
\end{prop}
\begin{proof} 
  Proposition~\ref{symmetric_space_time_diagram} is proven by induction on $t\in
  \mathcal{T}$.  Let $t = 0$. Property $ \widetilde{F}^t(x(0)) = S_i(x(t))$ is
  true because $x(0) = \overline{0}^i$.  Suppose that it is true for $t \in
  \mathcal{T}$.  Then, we have $\widetilde{F}^{t+1}(x(0)) =
  \widetilde{F}(\widetilde{F}^t(x(0))) = \widetilde{F}(S_i(x(t))$. By
  Lemma~\ref{symmetric_network}, $\widetilde{F}(S_i(x(t)) = S_i(F(x(t)) =
  S_i(x(t+1))$, which is the expected result.
\end{proof}

Remark that this result is due to the fact that $F$ and $\widetilde{F}$ are the
global transition functions of two symmetric $k$-\textsc{xor} circulant networks
that are isomorphic by definition (see
Figure~\ref{fig_xor-ca-17-24}). Proposition~\ref{symmetric_space_time_diagram}
implies that, for any automaton $i \in V$, the space-time diagram of
$(\overline{0}^{i}(t))_{t \in \mathcal{T}}$ is the symmetric space-time diagram
of $(\overline{0}^{M_i(t)})_{t \in \mathcal{T}}$ with respect to $i$ and is
related to the trace of automaton $i$. Thus, the space-time diagrams of
configurations of density $\frac{1}{n}$ carry information on the global
behaviours of $N$. We examine further these properties in the following results.
\begin{prop}
  \label{prop_densite-1} Let $N$ be a $k$-\textsc{xor} circulant network of size
  $n$ with automata set $V$ and global transition function $F$. The maximum
  convergence time, \emph{i.e.}, the maximal transient trajectory length, is
  reached by configurations of density $\frac{1}{n}$. Moreover, let $p_\ast$ be
  the period of the attractors reached by configurations of density
  $\frac{1}{n}$.  Then, for any configuration $x$ of $N$, the period of its
  attractor (\emph{i.e.}, of the attractor that is reached by the network when
  it is initially in configuration $x$) divides $p_\ast$.
\end{prop}
\begin{proof}
  Since all configurations of density $\frac{1}{n}$ are cyclic permutations of
  one another, by Proposition~\ref{prop_basic}.4 they all have isomorphic
  trajectories so that they all hit their limit set at the same time $t_\ast$
  and they all have the same period $p_\ast$. Now, consider configuration $x$
  and automaton $i$. By Proposition~\ref{symmetric_space_time_diagram}, the
  space-time diagram of $(\overline{0}^{M_i(t)})_{t \in \mathcal{T}}$ is the
  symmetric space-time diagram of $(\overline{0}^i(t))_{t \in \mathcal{T}}$ with
  respect to $i$. Thus, the space-time diagram of $(\overline{0}^{M_i(t)})_{t
    \in \mathcal{T}}$ hits its limit set at time $t_\ast$ and its period is
  $p_\ast$. This means that, $\forall i \in N$, the trace of automaton $i$ has
  period $p_\ast$ and hits its limit before $t_\ast$. Thus, the trajectory of
  $x$ reaches its limit set before $t_\ast$ and its period divides $p_\ast$.
\end{proof}

\subsubsection{$2$-\textsc{xor} circulant networks}
\label{sec-twoxor}

Let us now concentrate on $2$-\textsc{xor} circulant networks of arbitrary size
$n$ and pay particular attention to the space-time diagrams of configurations of
density $\frac{1}{n}$. We define the \emph{interaction-step} of such a network
$N$ as the smallest integer $s \neq 1 <n$ such that $\forall i \in V,$
\mbox{$(i, i+s~[n]) \in A$}.  As illustrated in Figure~\ref{fig_xor-ca-17-24}
$(a)$ and $(b)$, when $s = 0$ the space-time diagram is the Sierpinski
triangle. For other values of $s$, space-time diagrams seem like deformed
Sierpinski triangles. From these observations results the following lemma.
\begin{figure}[t!]
  \centerline{
    \begin{tabular}{ccc}
      \includegraphics[scale=0.8]{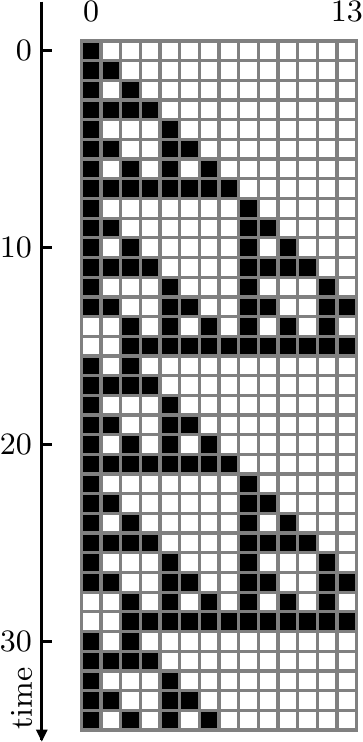}~ &
      ~~\includegraphics[scale=0.8]{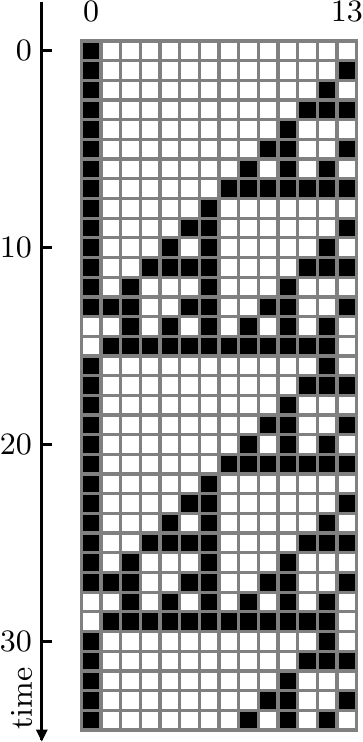}~ &
      ~~\includegraphics[scale=0.8]{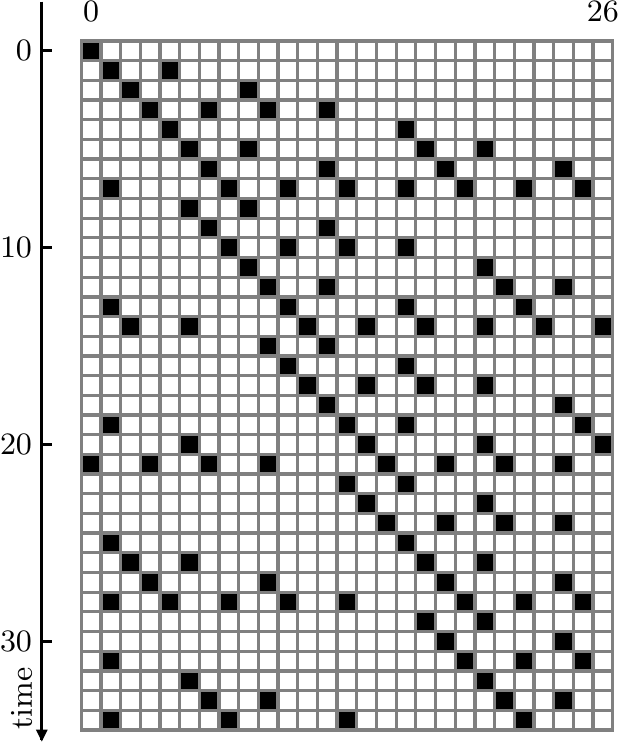}\\
      $(a)$ & $(b)$ & $(c)$
    \end{tabular}
  }
  \caption{Space-time diagrams $(a)$ of a $2$-\textsc{xor} circulant network of
    size $14$ and interaction-step $s=0$ (see Section~\ref{sec-twoxor}), $(b)$
    of its symmetric network and $(c)$ of another $2$-\textsc{xor} circulant
    network of size $27$ and interaction-step $4$.}
  \label{fig_xor-ca-17-24}
\end{figure}
\begin{lem}
  \label{lem_local}
  Let $N$ be a $2$-\textsc{xor} circulant network of size $n$ with
  interaction-step $s = 0$. The following holds:
  \begin{equation*}
    \forall i \in V, \forall q \in \mathbb{N},\ x_i(2^q) = x_{(i-2^q)~[n]}(0) 
    \oplus x_i(0) \text{.}
  \end{equation*}
\end{lem}
\begin{proof} 
  Lemma~\ref{lem_local} is proven by induction on $q$.\smallskip

  Let $i \in V$ be an arbitrary automaton and let $q$ equal $1$ initially. Then,
  obviously, the following is true:
  \begin{equation*}
    \begin{split}
      x_i(2) & = x_{(i-1)~[n]}(1) \oplus x_i(1)\\
      & = x_{(i-2)~[n]}(0) \oplus x_{(i-1)~[n]}(0) \oplus x_{(i-1)~[n]}(0)
      \oplus x_{i}(0)\\
      & = x_{(i-2)~[n]}(0) \oplus x_{i}(0) \text{,}
    \end{split}
  \end{equation*}
  and the basis of the induction holds.\smallskip

  Now, let us assume as induction hypothesis that $x_i(2^q) = x_{(i-2^q)~[n]}(0)
  \oplus x_i(0)$ is true for $q\in\mathbb{N}$.  In the sequel, we pay particular
  attention to states
  \begin{gather*}
    a = x_i(0)\text{,} \quad b = x_{(i-2^{q-1})~[n]}(0)\text{,}\quad c =
    x_{(i-2^q)~[n]}(0)\text{,} \quad d = x_i(2^{q-1})\text{,}\\
    e = x_{(i-2^{q-1})~[n]}(2^{q-1}) \quad \text{and} \quad f =
    x_i(2^{q})\text{,}
  \end{gather*}
    as illustrated in Figure~\ref{fig_xor-ca}.\smallskip
  \begin{figure}[t!]
    \centerline{\includegraphics[scale=0.8]{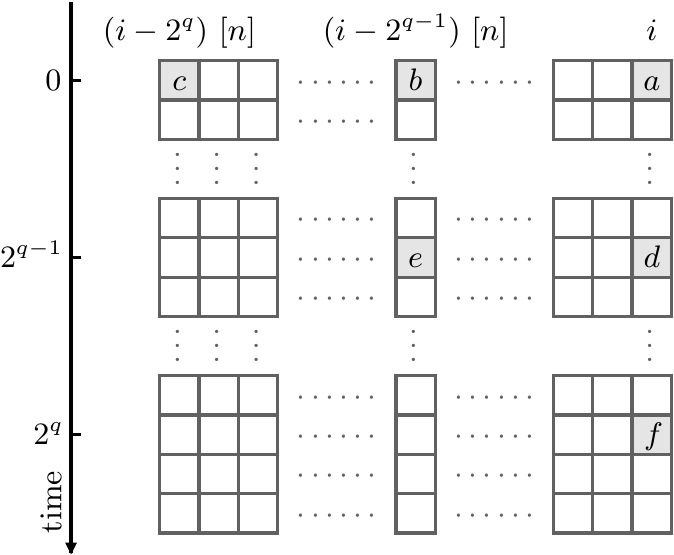}}
    \caption{Space-time diagram of an arbitrary $2$-\textsc{xor} circulant 
      network of size $n$ and interaction-step $s=0$.}
    \label{fig_xor-ca}
  \end{figure}
  
  Then, for $q+1$, according to the induction hypothesis, we have:
  \begin{equation*}
    d = a \oplus b \text{,\quad} e = b \oplus c \text{\quad and\quad} 
    f = d \oplus e\text{.}
  \end{equation*}
  Then, we derive that:
  \begin{equation*}
    f\ =\ d \oplus e\ =\ (a \oplus b) \oplus (b \oplus c)\ =\ a \oplus c 
    \text{.}
  \end{equation*}
  As a result, we can write:
  \begin{equation*}
    \forall i \in V, \forall q \in \mathcal{N},\ x_i(2^q) = x_i(0)
    \oplus x_{(i-2^q)~[n]}(0) \text{,}
  \end{equation*}
  which is the expected result.
\end{proof}
We will use this lemma to analyse $2$-\textsc{xor} circulant networks of size $n
= 2^p$, $p \in \mathbb{N^*}$, and interaction step $s=0$.

\subsubsection{$2$-\textsc{xor} circulant networks of sizes powers of $2$}

In this paragraph, we focus on $2$-\textsc{xor} circulant networks of sizes $n =
2^p$, where $p \in \mathbb{N^*}$. Let $x = (x_0, \ldots, x_{n-1}) \in \{0,1\}^n$
be a configuration of such a network $N$. We can see $x$ as the concatenation of
two vectors of sizes $\frac{n}{2}$ such that $x = (x',x'')$, where $x' = (x_0,
\ldots, x_{\frac{n}{2}-1})$ and $x'' = (x_{\frac{n}{2}}, \ldots, x_{n-1})$ are
called \emph{semi-configurations} of $x$. The \emph{repetition degree}
$\delta_r(x)$ of $x$ is then defined as:
\begin{equation*}
  \delta_r(x = (x',x'')) = \begin{cases}
    0 & \text{if } x' \neq x'' \text{,}\\
    \delta & \text{if } (x' = x'') \land (\delta_r(x') = \delta - 1)\text{.} 
  \end{cases}
\end{equation*}
Notice that if $x = (x',x')$, $x$ is said to be a \emph{repeated configuration}
and that, in the worst case ({\it i.e.}, when the repetition degree $\delta_r(x) =
\log_2(n)$ is maximal), the time complexity of the computation of the repetition
degree of any configuration $x$ equals $n$.\medskip

Proposition~\ref{prop_repeated_configuration} below characterises the dynamical
behaviours of repeated configurations $x \in \{0,1\}^n$ of repetition degree
$\delta_r(x) \geq \log_2(n) - 1$ in arbitrary $2$-\textsc{xor} circulant
networks of size $n = 2^p$,  $p \in \mathbb{N}^*$.
\begin{prop}
  \label{prop_repeated_configuration}
  Let $N$ be a $2$-\textsc{xor} circulant network of size $n = 2^p$, $p \in
  \mathbb{N}^*$, and interaction-step $s$. Configurations $x \in \{0,1\}^n$
  of repetition degree $\delta_r(x) \geq \log_2(n) - 1$ converge towards $(0,
  \ldots, 0)$ in no more than $2$ time steps.
\end{prop}
\begin{proof}
  First, notice that because $N$ is a $2$-\textsc{xor} circulant network of size
  $n = 2^p$, $p \in \mathbb{N}^*$, there exist only $4$ repeated
  configurations of degree no smaller than $\log_2(n) - 1$, that is, $(0, 1,
  \ldots, 0, 1)$, its dual $(1, 0, \ldots, 1, 0)$ and $(1, \ldots, 1)$ and its
  dual $(0, \ldots, 0)$. Let us consider the two distinct parities of $s$
  independently. Also, let $t \in \mathcal{T}$ and let $x(t)$ be either
    $(0, 1, \ldots, 0, 1)$ or $(1, 0, \ldots, 1, 0)$.
  \begin{enumerate}
  \item If $s$ is even, then, by hypothesis on $x(t)$:
    \begin{equation*}
      \forall i \in V,\ x_{(i+s)~[n]}(t+1) = x_i(t) \oplus x_{(i+s-1)~[n]}(t) = 1 
      \text{.}
    \end{equation*}
  \item If $s$ is odd, then, by hypothesis on $x(t)$:
    \begin{equation*}
      \forall i \in V,\ x_{(i+s)~[n]}(t+1) = x_i(t) \oplus x_{(i+s-1)~[n]}(t) = 0 
      \text{.}
    \end{equation*}
  \end{enumerate}
  This, together with Propositions~\ref{prop_basic}.2 and~\ref{prop_basic}.3,
  yields the expected result.
\end{proof}

From now on, we restrict the study to $2$-\textsc{xor} circulant networks of
sizes $n = 2^p$, $p \in \mathbb{N}^*$, and interaction-steps $s = 0$. We show
that such networks necessarily converge towards configuration $(0, \ldots, 0)$
in no more than $n$ time steps and that initial configurations with an odd
number of $1$ converge in exactly $n$ steps.
\begin{thm}
  \label{thm_convergence-lineaire}
  Let $N$ be a $2$-\textsc{xor} circulant network of size $n = 2^p$,  $p
  \in \mathbb{N}^*$, and interaction-step $0$. Any configuration $x$
  converges to the stable configuration $(0, \dots, 0)$ in no more than $n$ time
  steps.
\end{thm}
\begin{proof}
  Since $n = 2^p$, by Lemma~\ref{lem_local}, we directly draw:
  \begin{equation*}
    \forall i \in V,\ x_i(n) = x_i(0) \oplus x_{i+n~[n]}(0) = x_i(0) \oplus x_i(0)
    = 0 \text{.}
  \end{equation*}
  This allows to conclude that any configuration $x$ converges to the stable
  configuration $(0, \dots, 0)$ in no more than $n$ time steps.
\end{proof}
Now, let us consider the configurations for which the convergence time is
maximal.
\begin{lem}
  \label{lem_repeated}
  Let $N$ and $N'$ be two $2$-\textsc{xor} circulant networks of respective
  sizes $n = 2^{p+1}$ and $n' = 2^p$, $p \in \mathbb{N}^*$, and
  interaction-steps $0$. Let $x'$ be a configuration of size $2^p$ and $x = (x',
  x')$ be a repeated configuration of size $2^{p+1}$. Then, for any $t \in
  \mathcal{T}$, $x(t) = (x'(t), x'(t))$.
\end{lem}
\begin{proof}
  Considering an arbitrary repeated configuration $x$ of $N$, by induction on $t$, we show that
  $\forall t \in \mathcal{T},\ x(t) = (x'(t), x'(t))$. Let
  us denote by $G' = (V', A')$ the interaction graph of $N'$.\smallskip
  
  By hypothesis, the proposition is true for $t = 0$.\smallskip

  Now, consider that $x(t) = (x'(t), x'(t))$ for $t \in \mathcal{T}$ and that
  \begin{equation*}
    \forall i \in V,\ x_i(t + 1) = x_{(i-1)~[n]}(t) \oplus x_{i}(t) \text{.}
  \end{equation*}
  Since $x(t)$ is a repeated configuration, we have:
  \begin{equation*}
    \begin{split}
      \forall i \in V,\ x_i(t + 1) &
      = x_{(i-1)~[n]}(t) \oplus x_{i}(t)\\
      & = x_{(i-1+2^p)~[n]}(t) \oplus x_{(i+2^p)~[n]}(t)\\
      & = x_{(i+2^{p})~[n]}(t+1) \text{.}
    \end{split}
  \end{equation*}
  Thus, $x(t+1)$ is also repeated and it satisfies:
  \begin{equation*}
    \begin{split}
      \forall i \in V', x_i(t+1) &
      = x_{(i-1)~[n']}(t) \oplus x_{i}(t)\\
      & = x'_{(i-1)~[n']}(t) \oplus x'_i(t)\\
      & = x'_i(t+1) \text{.}
    \end{split}
  \end{equation*}
  As a result, it holds that $x(t+1) = (x'(t+1), x'(t+1))$.
\end{proof}
\begin{prop}
  \label{prop_convergence-lineaire-impair}
  Let $N$ be a $2$-\textsc{xor} circulant network of size $n = 2^p$, $p \in
  \mathbb{N}^*$, and interaction-step $0$. Any configuration $x$ such that $n
  \cdot d(x) \equiv 1~[n]$ (with an odd number of $1$s) converges in $n$ time
  steps exactly.
\end{prop}
\begin{proof}Proposition~\ref{prop_convergence-lineaire-impair} is proven by induction on $p$.\smallskip


  If $p = 1$, according to Propositions~\ref{prop_basic}.3
  and~\ref{prop_repeated_configuration}, configurations of repetition degree
  $\log_2(n) - 1$ are proven to converge in $2$ time steps. Thus,
  Proposition~\ref{prop_convergence-lineaire-impair} holds for $p =
  1$.\smallskip

  Suppose that for $p = q$, any configuration $x$ such that $2^q \cdot
  d(x) \equiv 1~[n]$ converges in $2^q$ time steps.\smallskip

  Now, suppose that $p = q + 1$ and consider a $2$-\textsc{xor} circulant 
  network $N$ of size $n = 2^{q+1}$ and interaction-step $0$. Let  $x$ be a
  configuration of size $2^{q+1}$ such that $n \cdot d(x) \equiv 1~[n]$. We
  show that after $2^q$ time steps:
  \begin{enumerate}
  \item \emph{$x(2^q)$ is a repeated configuration of the form $x(2^q) =
      (x'(2^q),x'(2^q))$}. By Lemma~\ref{lem_local}, $\forall i \in \{0, \ldots,
    2^q-1\},\ x_i(2^q) = x_i(0) \oplus x_{(i+2^{q})~[n]}(0)$. Hence:
    \begin{multline*}
      \forall i \in \{0, \ldots, 2^q-1\},\\
      x_i(2^q) = x_{(i+2^{q+1})~[n]}(0) \oplus x_{(i+2^{q})~[n]}(0) = 
      x_{(i+2^{q})~[n]}(2^q) \text{.}
    \end{multline*}
  \item \emph{$x'$ has an odd number of $1$s}. By Lemma~\ref{lem_local} and the
    property above, since $\forall i \in \{0, \ldots 2^q-1\}, x'_i(2^q) =
    x_i(2^q) = x_i(0) + x_{(i+2^{q})~[n]}(0)$, each automaton of $x(0)$
    influences exactly one automaton of $x'$. If $x'_i(2^q) = 0$, then the
    states of both the automata of $x(0)$ that influence $x'_i(2^q)$ must have
    the same parity. If $x'_i(2^q) = 1$ then the states of both the automata of
    $x(0)$ that influence $x'_i(2^q)$ must have different parities. Since
    there is an odd number of $1$s in $x(0)$, there is an odd number of $1$s in
    $x'(2^q)$.
  \end{enumerate}
  By Lemma~\ref{lem_repeated}, $x(2^q)$ behaves exactly like $x'(2^q)$. By the
  induction hypothesis, $x'$ converges in exactly $2^q$ time steps. Thus $x$
  converges in exactly $n = 2^{q+1}$ time steps.
\end{proof}

\section{Conclusion and perspectives}
\label{sec_conclusion}

With this study, we have endeavoured to show that non-monotony is an interesting
concept {\it per se}, despite the lack of specific attention it has received so
far. On the one hand, to serve as a stepping-stone and acquire some initial
intuitions in this domain, we have considered a special family of non-monotone
Boolean automata networks that we named \textsc{xor} circulant networks. In
particular, we have focused on the trajectorial and asymptotic behaviours of
these networks, considered their convergence times and characterised their
attractors. Globally, this preliminary formal analysis revealed that simple
non-monotone networks can exhibit non-trivial, engaging properties. On the other
hand, more generally and informally, we also have put forward several arguments
to support the idea that work needs to be done to build a better understanding
of the role of non-monotony in the behaviours of automata networks.  In these
lines, we have mentioned that studies in this context could find concrete and
relevant applications in biology and in particular in the modelling of genetic
regulation networks by automata networks. In addition, we have given two
theoretical arguments in favour of our insights by which non-monotony could be
responsible for singular network behaviours.  First,
exploiting~\cite{Delaplace2011}, we have argued that non-monotony may be
responsible for the strongly connected components of networks being
non-separable, minimal functional modules.
Second, in Proposition~\ref{prop_asynchronous-general}, with the state
transition systems formalism, we have considered ``synchronism sensitivity'',
that is, the property of Boolean automata networks to display significant
behavioural changes when synchronism is added to their automata state
updates. And in this context, we have shown that the smallest synchronism
sensitive Boolean automata networks are also non-monotone.  \medskip

The issues presented in this paper open many research directions that could help
develop a better understanding of the precise role of non-monotony in formal
automata networks and, {\it a fortiori}, in real biological regulation networks.
One of these perspectives consists in identifying the relations that exist
between monotone and non-monotone Boolean automata networks.
In~\cite{Noual2011a}, some preliminary results are derived on synchronism
sensitivity. In particular it is shown that this property requires specific
circuits underlying the networks structures. Also, monotone examples of
synchronism sensitive networks are given. The interesting point is that all of
them seem to involve a \emph{monotone coding} of non-monotony. With
Proposition~\ref{prop_asynchronous-general}, this naturally raises the question
of whether non-monotony (taken in a more general sense than what we did formally
above) can account in a certain way for the synchronism sensitivity in arbitrary
monotone and non-monotone networks.  Thus,
Proposition~\ref{prop_asynchronous-general} together with the work presented
in~\cite{Noual2011a} call for further researches in this direction.  With
sufficient knowledge in this context, we then hope to move on to the subject of
modularity as developed in~\cite{Delaplace2011} and work on establishing the
exact non-separability conditions of strongly connected networks. In this
context, the first important questions that need to be addressed are: ``Does
there exist monotone strongly connected networks that are separable into
functional modules?'' and ``How does non-monotony relate to the non-separability of
non-monotone networks?''. The relevance of these questions lies in that their
answers will help understand modularity in biological regulation networks, which
is a central issue in present biological research frameworks such as synthetic
biology. Eventually, further analyses also need to be done on the dynamical
behaviours of \textsc{xor} circulant networks. Indeed, we believe that these
networks constitute very promising instances of non-monotone networks because of
their apparent simplicity and because, since they involve underlying structural
circuits, their dynamical behaviours are potentially diverse and complex. Thus,
pursuing in this direction, we hope to obtain generalisations of the results
that figure above concerning the parallel updating mode by relaxing structural
constraints step by step. Also, another interesting perspective in this
framework is to consider \textsc{xor} circulant networks as state transition
systems, under the asynchronous and general updating modes. This perspective is
motivated in particular by the fact that, according to
Proposition~\ref{prop_asynchronous-general}, the smallest synchronism sensitive
networks are either \textsc{xor} circulant networks of size $2$ and
interaction-step $0$, or networks that have the same structures as these and
comparable non-monotone interactions.

\section{Acknowledgements}
\label{sec_acknowledgments}

We thank Florian Rabin for his relevant comments and are indebted to the
\emph{Agence nationale de la recherche} and the \emph{R{\'e}seau national des
  syst{\`e}mes complexes} that have respectively supported this work through the
projects Synbiotic (\textsc{anr} 2010 \textsc{blan} 0307\,01) and
M{\'e}t{\'e}ding (\textsc{rnsc} \textsc{ai}10/11-\textsc{l}03908).

\bibliographystyle{elsarticle-num}
\bibliography{non_monotony} 

\begin{thebibliography}{10}
\expandafter\ifx\csname url\endcsname\relax
  \def\url#1{\texttt{#1}}\fi
\expandafter\ifx\csname urlprefix\endcsname\relax\def\urlprefix{URL }\fi
\expandafter\ifx\csname href\endcsname\relax
  \def\href#1#2{#2} \def\path#1{#1}\fi

\bibitem{McCulloch1943}
W.~S. McCulloch, W.~Pitts, {A logical calculus of the ideas immanent in nervous
  activity}, Journal of Mathematical Biology 5 (1943) 115--133.

\bibitem{Kauffman1969a}
S.~A. Kauffman, {Metabolic stability and epigenesis in randomly constructed
  genetic nets}, Journal of Theoretical Biology 22 (1969) 437--467.

\bibitem{Kauffman1969b}
S.~A. Kauffman, {Homeostasis and differentiation in random genetic control
  networks}, Nature 224 (1969) 177--178.

\bibitem{Hopfield1982}
J.~J. Hopfield, {Neural networks and physical systems with emergent collective
  computational abilities}, Proceedings of the National Academy of Sciences of
  the USA 79 (1982) 2554--2558.

\bibitem{Hopfield1984}
J.~J. Hopfield, {Neurons with graded response have collective computational
  properties like those of two-state neurons}, Proceedings of the National
  Academy of Sciences of the USA 81 (1984) 3088--3092.

\bibitem{Kauffman1971}
S.~A. Kauffman, {Current topics in developmental biology}, Vol.~6, Elsevier,
  1971, Ch. {Gene regulation networks: a theory for their global structures and
  behaviors}, pp. 145--181.

\bibitem{Kauffman1993}
S.~A. Kauffman, The origins of order, Oxford University Press, 1993.

\bibitem{Jacob1961a}
F.~Jacob, J.~Monod, {Genetic regulatory mechanisms in the synthesis of
  proteins}, Journal of Molecular Biology 3 (1961) 318--356.

\bibitem{Jacob1961b}
F.~Jacob, J.~Monod, {On the regulation of gene activity}, Cold Spring Harbor
  Symposia on Quantitative Biology 26 (1961) 193--211.

\bibitem{Thomas1973}
R.~Thomas, {Boolean formalisation of genetic control circuits}, Journal of
  Theoretical Biology 42 (1973) 563--585.

\bibitem{Thomas1981}
R.~Thomas, {On the relation between the logical structure of systems and their
  ability to generate multiple steady states or sustained oscillations}, in:
  Numerical methods in the study of critical phenomena, Vol.~9 of Springer
  Series in Synergetics, Springer, 1981, pp. 180--193.

\bibitem{Thomas1991}
R.~Thomas, {Regulatory networks seen as asynchronous automata: a logical
  description}, Journal of Theoretical Biology 153 (1991) 1--23.

\bibitem{Floreen1989}
P.~Flor{\'e}en, P.~Orponen, {On the computational complexity of analyzing
  Hopfield nets}, Complex Systems 3 (1989) 577--587.

\bibitem{Cosnard1992}
M.~Cosnard, P.~Koiran, H.~Paugam-Moisy, {Complexity issues in neural network
  computations}, in: Proceedings of LATIN, Vol. 583 of Lecture Notes in
  Computer Science, Springer, 1992, pp. 530--544.

\bibitem{Koiran1993}
P.~Koiran, {Puissance de calcul des réseaux de neurones artificiels}, Ph.D.
  thesis, {\'E}cole normale supérieure de Lyon (1993).

\bibitem{Orponen1997}
P.~Orponen, {Computing with truly asynchronous threshold logic networks},
  Theoretical Computer Science 174 (1997) 123--136.

\bibitem{Gajardo2002}
A.~Gajardo, A.~Moreira, E.~Goles, {Complexity of Langton's ant}, Discrete
  Applied Mathematics 117 (2002) 41--50.

\bibitem{Robert1986}
F.~Robert, {Discrete iterations: a metric study}, Vol.~6 of Springer Series in
  Computational Mathematics, Springer, 1986.

\bibitem{Goles1981}
E.~Goles, J.~Olivos, {Comportement p{\'e}riodique des fonctions {\`a} seuil
  binaires et applications}, Discrete Applied Mathematics 3 (1981) 93--105.

\bibitem{Goles1990}
E.~Goles, {Neural and automata networks: dynamical behaviour and applications},
  Mathematics and Its applications, Kluwer Academic Publishers, 1990.

\bibitem{Aracena2004a}
J.~Aracena, J.~Demongeot, E.~Goles, {Fixed points and maximal independent sets
  in \textsc{and--or} networks}, Discrete Applied Mathematics 138 (2004)
  277--288.

\bibitem{Aracena2004b}
J.~Aracena, J.~Demongeot, E.~Goles, {On limit cycles of monotone functions with
  symmetric connection graph}, Theoretical Computer Science 322 (2004)
  237--244.

\bibitem{Remy2008}
E.~Remy, P.~Ruet, D.~Thieffry, {Graphic requirements for multistability and
  attractive cycles in a Boolean dynamical framework}, Advances in Applied
  Mathematics 41 (2008) 335--350.

\bibitem{Demongeot2010}
J.~Demongeot, E.~Goles, M.~Morvan, M.~Noual, S.~Sen{\'e}, {Attraction basins as
  gauges of robustness against boundary conditions in biological complex
  systems}, PLoS One 5 (2010) e11793.

\bibitem{Demongeot2011}
J.~Demongeot, M.~Noual, S.~Sené, {Combinatorics of Boolean automata circuits
  dynamics}, Discrete Applied Mathematics,~in press.

\bibitem{Richard2011}
A.~Richard, {Local negative circuits and fixed points in non-expansive Boolean
  networks}, Discrete Applied Mathematics 159 (2011) 1085--1093.

\bibitem{Goles1985}
E.~Goles-Chacc, F.~Fogelman-Soulie, D.~Pellegrin, {Decreasing energy functions
  as a tool for studying threshold networks}, Discrete Applied Mathematics 12
  (1985) 261--277.

\bibitem{Cosnard1997}
M.~Cosnard, E.~Goles, {Discrete state neural networks and energies}, Neural
  Networks 10 (1997) 327--334.

\bibitem{Remy2003}
E.~Remy, B.~Moss\'e, C.~Chaouiya, D.~Thieffry, {A description of dynamical
  graphs associated to elementary regulatory circuits}, Bioinformatics 19
  (2003) ii172--ii178.

\bibitem{Chaouiya2004}
C.~Chaouiya, E.~Remy, P.~Ruet, D.~Thieffry, {Qualitative modelling of genetic
  networks: from logical regulatory graphs to standard Petri nets}, in:
  Proceedings of Petri Nets, Vol. 3099 of Lecture Notes in Computer Science,
  Springer, 2004, pp. 137--156.

\bibitem{Colon-Reyes2005}
O.~Col{\'o}n-Reyes, R.~Laubenbacher, B.~Pareigis, {Boolean monomial dynamical
  systems}, Annals of Combinatorics 8 (2005) 425--439.

\bibitem{Jarrah2010}
A.~S. Jarrah, R.~Laubenbacher, A.~Veliz-Cuba, {The dynamics of conjunctive and
  disjunctive Boolean network models}, Bulletin of Mathematical Biology 72
  (2010) 1425--1447.

\bibitem{Mendoza1998}
L.~Mendoza, E.~Alvarez-Buylla, {Dynamics of the genetic regulatory network for
  \textit{Arabidopsis thaliana} flower morphogenesis}, Journal of Theoretical
  Biology 193~(2) (1998) 307--319.

\bibitem{Mendoza1999}
L.~Mendoza, D.~Thieffry, E.~Alvarez-Buylla, {Genetic control of flower
  morphogenesis in \textit{Arabidopsis thaliana}: a logical analysis},
  Bioinformatics 15 (1999) 593--606.

\bibitem{Aracena2006}
J.~Aracena, M.~Gonz{\'a}lez, A.~Zu{\~n}iga, M.~A. Mendez, V.~Cambiazo,
  {Regulatory network for cell shape changes during drosophila ventral furrow
  formation}, Journal of Theoretical Biology 239 (2006) 49--62.

\bibitem{Georgescu2008}
C.~Georgescu, W.~J.~R. Longabaugh, D.~D. Scripture-Adams, E.~S. David-Fung,
  M.~A. Yui, M.~A. Zarnegar, H.~Bolouri, E.~V. Rothenberg, {A gene regulatory
  network armature for T lymphocyte specification}, Proceedings of the National
  Academy of Sciences of the USA 105 (2008) 20100--20105.

\bibitem{Mendoza2010}
L.~Mendoza, F.~Pardon, {A robust model to describe the differentiation of
  T-helper cells}, Theory in Biosciences 129 (2010) 283--293.

\bibitem{Cull1971}
P.~Cull, {Linear analysis of switching nets}, Biological Cybernetics 8 (1971)
  31--39.

\bibitem{Huffman1956}
D.~A. Huffman, {Information theory}, Academic Press, 1956, Ch. {The synthesis
  of linear sequential coding networks}.

\bibitem{Elspas1959}
B.~Elspas, {The Theory of Autonomous Linear Sequential Networks}, IRE
  Transactions on Circuit Theory 6 (1959) 45--60.

\bibitem{Snoussi1980}
E.~H. Snoussi, {Structure et comportement it{\'e}ratif de certains modèles
  discrets}, Ph.D. thesis, Université Grenoble 1 -- Joseph Fourier (1980).

\bibitem{Choffrut1988}
C.~Choffrut (Ed.), {Automata networks}, Vol. 316 of Lecture Notes in Computer
  Science, Springer, 1988.

\bibitem{Noual2011a}
M.~Noual, {Synchronism vs asynchronism in Boolean networks}, Tech. rep.,
  \'Ecole normale sup{\'e}rieure de Lyon, arXiv:1104.4039 (2011).

\bibitem{Noual2011b}
M.~Noual, {General transition graphs and Boolean automata circuits}, submitted.

\bibitem{Noual2011c}
M.~Noual, S.~Sen{\'e}, {Towards a theory of modelling with Boolean automata
  networks - I. Theorisation and observations}, Tech. rep., \'Ecole normale
  sup{\'e}rieure de Lyon and Universit{\'e} d'{\'E}vry -- Val d'Essonne,
  arXiv:1111.2077 (2011).

\bibitem{Richard2004}
A.~Richard, J.-P. Comet, G.~Bernot, {R. Thomas' modeling of biological
  regulatory networks: introduction of singular states in the qualitative
  dynamics}, Fundamenta Informaticae 65 (2004) 373--392.

\bibitem{Richard2007}
A.~Richard, J.-P. Comet, {Necessary conditions for multistationarity in
  discrete dynamical systems}, Discrete Applied Mathematics 155 (2007)
  2403--2413.

\bibitem{Demongeot2008}
J.~Demongeot, A.~Elena, S.~Sen{\'e}, {Robustness in regulatory networks: a
  multi-disciplinary approach}, Acta Biotheoretica 56 (2008) 27--49.

\bibitem{Elena2008}
A.~Elena, J.~Demongeot, {Interaction motifs in regulatory networks and
  structural robustness}, in: Proceedings of CISIS, IEEE, 2008, pp. 682--686.

\bibitem{Aracena2009}
J.~Aracena, E.~Goles, A.~Moreira, L.~Salinas, {On the robustness of update
  schedules in Boolean networks}, Biosystems 97 (2009) 1--8.

\bibitem{Aracena2010}
J.~Aracena, {\'E}.~Fanchon, M.~Montalva, M.~Noual, {Combinatorics on update
  digraphs in Boolean networks}, Discrete Applied Mathematics 159 (2010)
  401--409.

\bibitem{Goles2010}
E.~Goles, M.~Noual, {Block-sequential update schedules and Boolean automata
  circuits}, in: Proceedings of Automata 2010, DMTCS, 2010, pp. 41--50.

\bibitem{Goles2011}
E.~Goles, M.~Noual, {Disjunctive networks and update schedules}, Advances in
  Applied Mathematics,~in press.

\bibitem{Robert1995}
F.~Robert, {Les syst{\`e}mes dynamiques discrets}, Vol.~19 of Math{\'e}matiques
  \& Applications, Springer, 1995.

\bibitem{Goles2008}
E.~Goles, L.~Salinas, {Comparison between parallel and serial dynamics of
  Boolean networks}, Theoretical Computer Science 396 (2008) 247--253.

\bibitem{Lederberg1950}
E.~M. Lederberg, {Lysogenicity in \emph{Escherichia coli} strain K-12},
  Microbial Genetics Bulletin 1 (1950) 5--9.

\bibitem{Eisen1970}
H.~Eisen, P.~Brachet, L.~Pereira~da Silva, F.~Jacob, {Regulation of repressor
  expression in $\lambda$}, Proceedings of the National Academy of Sciences of
  the USA 66 (1970) 855--862.

\bibitem{Thieffry1995}
D.~Thieffry, R.~Thomas, {Dynamical behaviour of biological regulatory networks
  -- II. Immunity control in bacteriophage lambda}, Bulletin of Mathematical
  Biology 57 (1995) 277--297.

\bibitem{Delaplace2011}
F.~Delaplace, H.~Klaudel, T.~Melliti, S.~Sen{\'e}, {Analysis of modular
  organisation of interaction networks based on asymptotic dynamics},
  submitted.

\bibitem{Milo2002}
R.~Milo, S.~Shen-Orr, S.~Itzkovitz, et~al., {Network motifs: simple building
  blocks of complex networks}, Science 298 (2002) 824--827.

\bibitem{Rives2003}
A.~W. Rives, T.~Galitski, {Modular organization of cellular networks},
  Proceedings of the National Academy of Sciences of the USA 100 (2003)
  1128--1133.

\bibitem{Gagneur2004}
J.~Gagneur, R.~Krause, T.~Bouwmeester, et~al., {Modular decomposition of
  protein-protein interaction networks}, Genome Biology 5 (2004) R57.

\bibitem{Richard2010}
A.~Richard, {Negative circuits and sustained oscillations in asynchronous
  automata networks}, Advances in Applied Mathematics 44 (2010) 378--392.

\end{thebibliography}

\end{document}